%% file: main.tex
\documentclass[conference]{IEEEtran}
\usepackage{etoolbox}
\IEEEoverridecommandlockouts
\usepackage{cite}
\usepackage{amsmath, amsthm, amssymb,amsfonts, dsfont}
\usepackage{algorithmic}
\usepackage{graphicx}
\usepackage{textcomp}
\usepackage{xcolor}
\usepackage{soul}
\usepackage{tikz}
\usepackage{graphics}
\usepackage{pgfplots}
\usepackage{mathtools, cuted} 
\usepackage[ruled,vlined]{algorithm2e}
\usepackage{xcolor}

\usepgfplotslibrary{groupplots}
\usetikzlibrary{backgrounds,calc,shadings,shapes,arrows,arrows.meta,shapes.arrows,shapes.symbols,shadows,patterns, datavisualization, datavisualization.formats.functions, matrix}
\def\BibTeX{{\rm B\kern-.05em{\sc i\kern-.025em b}\kern-.08em
    T\kern-.1667em\lower.7ex\hbox{E}\kern-.125emX}}

\AfterEndPreamble{
  \mathchardef\mathcomma\mathcode`\,
  \mathcode`\,="8000 
}
{\catcode`,=\active
  \gdef,{\mathcomma\discretionary{}{}{}}
}

\allowdisplaybreaks 
\pgfplotsset{compat=1.14}
\begin{document}

\title{Priority Flow Admission and Routing in SDN: Exact and Heuristic Approaches
}

\author{\IEEEauthorblockN{Jorge L\'opez, Maxime Labonne, Claude Poletti, Dallal Belabed
\IEEEauthorblockA{
\textit{Airbus Defence and Space}\\
\'Elancourt, France \\
\{jorge.lopez-c, maxime.labonne, claude.poletti, dallal.belabed\}@airbus.com
}}}
\maketitle

\begin{abstract}
\input{abstract}
\end{abstract}


\newtheorem{proposition}{Proposition}
\newtheorem{definition}{Definition}

\section{Introduction}\label{sec:intro}
\input{introduction}

\section{Preliminaries}\label{sec:prelim}
\input{figures/TIKZ_network_drawing}
\input{preliminaries}

\section{Problem statement and exact solution}\label{sec:exact}
\input{exact}

\section{Heuristic solution}\label{sec:heuristic}
\input{heuristic}

\section{Experimental results}\label{sec:experim}
\input{experiments}

\section{Conclusion}\label{sec:conc}
\input{conclusion}

\section*{Acknowledgement} The authors would like to thank Jean-Baptiste Munier, Herve Fritsch and Marc Cartigny from Airbus Defence and Space for fruitful discussions. Likewise, the authors would like to thank Djamal Zeghlache from T\'el\'ecom SudParis / Institut Polytechnique de Paris for kindly providing the test platform. 

\bibliographystyle{IEEETran}
\bibliography{references}

\appendices
\section{ILP Program Example}\label{sec:ap_ilp_prog}
\input{ilp_example}

\end{document}

%% file: abstract.tex
This paper proposes a novel admission and routing scheme which takes into account arbitrarily assigned priorities for network flows. The presented approach leverages the centralized Software Defined Networking~(SDN) capabilities in order to do so. Exact and heuristic approaches to the stated Priority Flow Admission and Routing~(PFAR) problem are provided. The exact approach which provides an optimal solution is based on Integer Linear Programming~(ILP). Given the potentially long running time required to find an exact and optimal solution, a heuristic approach is proposed; this approach is based on Genetic Algorithms~(GAs). In order to effectively estimate the performance of the proposed approaches, a simulator that is capable of generating semi-random network topologies and flows has been developed. Experimental results for large problem instances (up 50 network nodes and thousands of network flows), show that: i) an optimal solution can be often found in few seconds (even milliseconds), and ii) the heuristic approach yields close-to-optimal solutions (approximately 95\% of the optimal) in a fixed amount of time; these experimental results demonstrate the pertinence of the proposed approaches.

%% file: introduction.tex
In recent years, computer networks have become highly dynamic, and new solutions for their fast deployment and management have been proposed, providing extraordinary capabilities. Yet, these capabilities continue to be explored, and new possibilities remain open. Software Defined Networking~(SDN) \cite{sdn} is a network management technology, which enables fast reconfiguration through a centralized interface. SDN is currently deployed in large data centers, and mission critical networks (see for example \cite{googlesdn}). Moreover, taking advantage of the capabilities provided by SDN, novel functional network capabilities are proposed; for example, using load balancing for maximizing the resource utilization \cite{LBSDN} or maximizing the Quality of Experience~(QoE) of end-users \cite{httpQoESDN}. 

In computer networks where the resources (or medium) are shared, the usual strategy is to allocate certain \emph{dedicated} bandwidth to users (or applications), according to a predefined Service Level Agreement~(SLA). However, this scheme is not ideal for prioritizing important traffic, especially when the network is congested. Particularly, there is no notion of admission (or rejection) of passing traffic based on its priority. An interesting perspective is to take admission and possibly even routing decisions (drop or re-route) in order to privilege traffic with higher priority. In this paper, we aim at providing solutions for this issue by leveraging the centralized management capabilities of SDN. 

In order to better understand the problem we aim at solving, let us consider a simple example. Consider the data-plane depicted in Fig.~\ref{fig:ex_topo}, where the nodes in the graph represent network nodes that forward traffic, through the links between them, and the bandwidth capacity of each link is shown in the associated label (assume a predetermined bandwidth unit, e.g., megabytes per second, Mbps for short). Assume the flows traversing the network are listed in Table~\ref{tab:ex_flows}; furthermore, assume a larger value implies a higher priority. A solution is to assign the path $N_1\rightarrow N_2$ to the flow with ID 2 (and highest priority), the path $N_1\rightarrow N_3\rightarrow N_2$ to the flow with ID 4, the path $N_1\rightarrow N_4\rightarrow N_2$ to the flow with ID 1, and no path (drop) for the flow with ID 3 (and lowest priority). Ideally, all flows should be admitted; in order to do so, the problem can be reduced to finding the appropriate path for each flow. However, whenever this is not possible, flows with higher priority should be favored (i.e., flows with lower priority can be dropped). We note that in order to avoid a few (or a single) high priority flows (with large bandwidth requirements) from occupying all the available resources, we do not focus on a fully hierarchical admission of the flows. Rather, we focus on maximizing the total value of the priorities of the admitted flows. Nevertheless, a strict prioritization of network flows is possible, and we also discuss this possibility (see Section~\ref{sec:exact}).

\begin{figure}[!htb]
    \centering
    \input{figures/example_net_topo}
    \caption{Example network topology}
    \label{fig:ex_topo}
\end{figure}
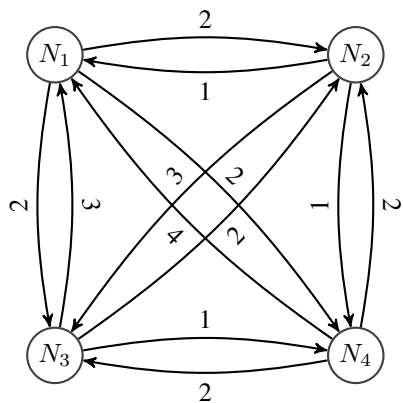 

\begin{table}[!htb]
	\centering
	\begin{tabular}{|c|c|c|c|c|}
		\hline
		\textbf{ID} & \textbf{source} 	& \textbf{destination} & \textbf{required bandwidth} & \textbf{priority}	\\ \hline
		1           & $N_1$         & $N_2$        & 2            & 10               		\\ \hline
		2           & $N_1$         & $N_2$        & 2            & 1000	             	\\ \hline
		3           & $N_3$         & $N_2$        & 1            & 1                 		\\ \hline
		4           & $N_1$         & $N_2$        & 2            & 100              		\\ \hline
	\end{tabular}
	\caption{Example flows}
	\label{tab:ex_flows}
\end{table}

We formally define the problem and propose an exact solution through integer linear programming (see Section~\ref{sec:exact} for the exact solution and Section~\ref{sec:prelim} for the related concepts). As we aim at providing a dynamic and runtime routing solution, searching for an exact optimal solution may be unfeasible in a reasonable time, for a large number of nodes and flows in the network. Thus, a heuristic approach based on evolutionary algorithms is proposed (see Section~\ref{sec:heuristic}). In order to estimate both the time feasibility of the exact approach and the optimality of the proposed heuristic approach, an experimental evaluation has been performed (see Section~\ref{sec:experim}). Our experiments show that the exact solution can be often found in a reasonable amount of time, and furthermore when this is not possible, our heuristic approach can provide close-to-optimal results in a fixed amount of time. Based on those experiments we conclude that both approaches are applicable and complementary to solve the stated problem.

%% file: figures/example_net_topo.tex
\begin{tikzpicture}[node distance=4cm,>=stealth',bend angle=45,auto]
    \tikzstyle{switch}=[circle,thick,draw=black!75, inner sep=0.075cm]
    \tikzstyle{host}=[circle,thick,draw=black!75,fill=black!80,text=white, inner sep=0.075cm]
    \tikzstyle{undirected}=[thick]
    \tikzstyle{directed}=[thick,->]

    \node[switch] (s1) {$N_1$};
    \node[switch, right of=s1] (s2) {$N_2$};
    \node[switch, below of=s1] (s3) {$N_3$};
    \node[switch, below of=s2] (s4) {$N_4$};
   
    \path   
            (s1)    edge[directed, bend left=10]    node[above] {2}   (s2)
                    edge[directed, bend right=10]    node[above, rotate=90] {2}   (s3)
                    edge[directed, bend left=10]    node[above, rotate=-45] {2}   (s4)
            (s2)    edge[directed, bend left=10]    node[below] {1}  (s1)
                    edge[directed, bend right=10]    node[above, rotate=45] {3}   (s3)
                    edge[directed, bend right=10]    node[above, rotate=90] {1}   (s4)
            (s3)    edge[directed, bend right=10]    node[above, rotate=-90] {3}   (s1)
                    edge[directed, bend right=10]    node[below, rotate=45] {2}   (s2)
                    edge[directed, bend left=10]    node[above] {1}   (s4)
            (s4)    edge[directed, bend left=10]    node[below, rotate=-45] {4}   (s1)
                    edge[directed, bend right=10]    node[above, rotate=-90] {2}   (s2)
                    edge[directed, bend left=10]    node[below] {2}   (s3);
\end{tikzpicture}

%% file: preliminaries.tex
In our attempt to make this paper as self-contained as possible, this section briefly reviews some of the necessary concepts used throughout the paper.

\subsection{Software Defined Networking}
In traditional networks, the configuration, management, and data-forwarding interfaces are distributed / located at each of the data forwarding devices (switches / routers) in the data-plane. The data-paths (the paths network packets follow in a data-plane) in the network are the result of the configuration on each of the forwarding devices; each of the devices has a local configuration and management interface. Thus, in order to re-configure the data-paths, several devices must be re-configured; as a consequence, while re-configuring each device the network may be in an inconsistent state, the process can be error-prone and slow. As an example, assume a data-plane in a traditional network as (only the data-plane) shown in Fig.~\ref{fig:sdn_ex}. Assume all flows from $h_1$ to $h_2$ follow the data-path depicted in solid arrows ($h_1\rightarrow r_1\rightarrow r_2\rightarrow h_2$); consider the link $(r_1,r_2)$ is too loaded. In order to re-configure some of the traffic to use an alternative data-path, for example $h_1\rightarrow r_1\rightarrow r_3\rightarrow r_2\rightarrow h_2$ (depicted in dashed arrows in Fig.~\ref{fig:sdn_ex}), the forwarding devices $r_1,r_3$, and $r_2$ must be re-configured, independently.

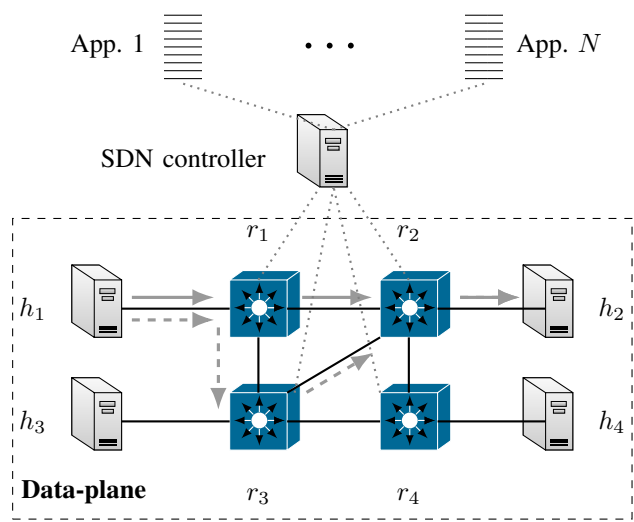
\begin{figure}[!htb]
    \centering
    \input{figures/example_sdn}
    \caption{Example SDN architecture}\label{fig:sdn_ex}
\end{figure}

SDN overcomes these limitations by separating the control and data-plane layers \cite{sdn}. With a centralized SDN controller, SDN applications can automatically re-configure the SDN data-plane in a timely manner. Furthermore, the devices in the data-plane may have different configuration protocols and interfaces (called the southbound interface of the SDN controller), while the SDN controller has a single communication protocol (northbound interface) with the applications; thus, simplifying communication with heterogeneous and vendor-agnostic data-planes. Finally, SDN-enabled forwarding devices steer (route / forward) the incoming network packets based on so-called flow rules installed by the SDN applications (through the controller). A flow rule consists of three main (functional) parts: a packet matching part, an action part and a location / priority part. The matching part describes the values which a received network packet should have for a given rule to be applied. The action part states the required operations to perform to the matched network packets, while the location / priority part controls the hierarchy of the rules using tables and priorities. Finally, it is important to note that our approach is generic to any centralized routing management paradigm. Nonetheless, as SDN is the only widespread technology that provides the desired capabilities, we focus on SDN as an enabler technology.

\subsection{Integer Linear Programming~(ILP)}
An ILP problem has the following general form \cite{copbook}:
\begin{equation}
    \begin{aligned}
        &{\text{maximize}}    & \mathbf{c}^{T}\mathbf{x} \\
        &{\text{subject to}}  & \mathbf{Ax} \leq \mathbf{b}\\
        &                   & \mathbf{x} \geq \mathbf{0}\\
        &                   & \mathbf{x} \in \mathds{Z}^{n} 
    \end{aligned}
\end{equation}
where $\mathbf{x}$ is a vector of $n$ non-negative integers, $\mathbf{c}$ and $\mathbf{b}$ are integer vectors, and $\mathbf{A}$ is an integer matrix. 
The function to maximize is the \emph{objective function} representing the objective (or goal) of the optimization. The remainder of the problem formulation presents the \emph{constraints}, i.e., what the maximization of the objective function is subject to. An ILP problem represents in most cases a combinatorial optimization problem; there exists a finite number of possible solutions, and the solutions must be integer-valued. A given instance of the optimization problem is the tuple $(F, c)$, where $F\subseteq \mathds{Z}^{n}$ is a domain of feasible points, and $c$ is an objective function such that $c:F\to \mathds{Z}_{+}$; hereafter, $\mathds{Z}_{+}$ represents the set of non-negative integers. The solution of the optimization problem is to find $f \in F$ such that $c(f) \geq c(y); \forall y \in F$; in this case, the point $f$ is a globally optimal solution. 

Algorithms to solve linear programming problems have been studied since the 1940s. One of most popular ones is the Dantzig's Simplex algorithm \cite{dantzig}. Correspondingly, to solve ILP problems, the cutting plane methods based on Gomory's algorithm \cite{gomory} can be applied. Modern tools for solving instances of ILP problems as Gurobi \cite{gurobi} or CPLEX \cite{cplex}, use such cutting plane methods, branch and bound algorithms \cite{BnB}, branch and cut \cite{BnC} algorithms and  others. It is important to mention that even if ILP is NP-hard \cite{sipser}, in practice, problems with hundreds or thousands of variables can be solved in a few seconds, due to a large number of implemented heuristics. 

\subsection{Genetic Algorithms}

Genetic algorithms~(GAs) as described by Goldberg \cite{goldbergGAs} are ``search procedures based on the mechanics of natural selection and natural genetics.'' In general, a genetic algorithm consists of the following steps:

\begin{enumerate}
	\item \emph{Initialization} -- Create a set of initial candidate solutions (often randomly generated), usually referred to as the population. Such solutions for ILP problems, candidate solutions (also called individuals) can be the values of the variables of the problem. In general, the individuals can be encoded as some primitive data-type, known as the genotype, and such solutions have an interpretation (for example, which paths to assign to which flows) called the phenotype.
	\item\label{algo:ga_select} \emph{Selection} -- From the candidate solutions, select the most fit; for an ILP problem that would be those with higher value of the objective function, otherwise a fitness function needs to be defined. Note that, often the selection is not only done based on the individual's fitness but, also based on \emph{how different} the individual is from the others; this is done to avoid the convergence of the population to the individual with highest fitness (not necessarily the optimal)
	\item \emph{Generation} -- The most fit individuals get to breed and produce new generations. Crossover (takes some parts of the \emph{parents} to create a new individual) and mutation (randomly changes some of the individual's values, not to have repeated exact populations) are the most common genetic operators. Crossover and mutation are applied to the population randomly, depending on a probability that may be set, in order to control how often such operators are applied. After the generation the process restarts with the selection step (step~\ref{algo:ga_select})
	\item \emph{Termination} -- For bounded optimization, the GA stops whenever the maximal value is obtained by the fitness of a given individual. Other possibilities are to stop whenever some number of generations / time is reached, or whenever a number of generations no longer produce better results (best maximal fitness).
\end{enumerate}

It is important to note that GAs are so-called meta heuristic algorithms, as such, they may produce a good quality solution or they may not, there is no guarantee that such approaches do so. Although the seemingly random nature of GAs may lead to believe that good results are not often obtained, the search is guided by the fitness function, solutions are preserved according to their level of fitness; GAs attempt to exploit the structure of good candidate solutions by recombining them. Therefore, in practice, GAs often provide good solutions for hard problems.

%% file: figures/example_sdn.tex
\begin{tikzpicture}
    \node[server]                               (h1)    {};
    \node               at ([xshift=-1cm]h1)    (h1l)   {$h_1$};
    \node[l3 switch]    at ([xshift=2cm]h1)     (s1)    {};
    \node               at ([yshift=1cm]s1)     (s1l)   {$r_1$};
    \node[l3 switch]    at ([xshift=2cm]s1)     (s2)    {};
    \node               at ([yshift=1cm]s2)     (s2l)   {$r_2$};
    \node[server]       at ([xshift=2cm]s2)     (h2)    {};
    \node               at ([xshift=.7cm]h2)    (h2l)   {$h_2$};
    \node[l3 switch]    at ([yshift=-1.5cm]s1)  (s3)    {};
    \node               at ([yshift=-1cm]s3)    (s3l)   {$r_3$};
    \node[server]       at ([xshift=-2cm]s3)    (h3)    {};
    \node               at ([xshift=-1cm]h3)    (h3l)   {$h_3$};
    \node[l3 switch]    at ([xshift=2cm]s3)     (s4)    {};
    \node               at ([yshift=-1cm]s4)    (s4l)   {$r_4$};
    \node[server]       at ([xshift=2cm]s4)     (h4)    {};
    \node               at ([xshift=.7cm]h4)    (h4l)   {$h_4$};
        
    \draw[thick] (h1)--(s1);
    \draw[thick] (h2)--(s2);
    \draw[thick] (h4)--(s4);
    \draw[thick] (h3)--(s3);
        
    \draw[thick] (s1.east)--(s2.west);
    \draw[thick] (s3.east)--(s4.west);
    \draw[thick] (s1.south)--(s3.north);
    \draw[thick] (s2.south)--(s4.north);
    \draw[thick] (s2.south west)--(s3.north east);
        
    \node[rectangle,draw,dashed,fill=none,text height=3.5cm,text width=8cm, minimum width=8cm,minimum height=4cm] at ([xshift=.85cm, yshift=-.8cm]s1) (main)  {\textbf{Data-plane}};
        
    \draw[-Latex,very thick, gray!80, shorten <=0.2cm,shorten >=0.2cm]([yshift=0.15cm,xshift=0.3cm]h1.east)-- ([yshift=0.15cm]s1.west);
    \draw[-Latex,very thick, gray!80, shorten <=0.1cm,shorten >=0.1cm]([yshift=0.15cm,xshift=0.1cm]s1.east)-- ([yshift=0.15cm]s2.west);
    \draw[-Latex,very thick, gray!80, shorten <=0.3cm,shorten >=0.7cm]([yshift=0.15cm]s2.east)-- ([yshift=0.15cm]h2.west);
        
    \draw[-Latex, dashed, very thick, gray!80, shorten <=0.2cm,shorten >=0.2cm]([yshift=-0.15cm,xshift=0.3cm]h1.east)-- ([yshift=-0.15cm]s1.west);
    \draw[-Latex,dashed, very thick, gray!80, shorten <=0.1cm,shorten >=0.3cm]([yshift=-0.15cm, xshift=-0.15cm]s1.west)-- ([yshift=-0.15cm, xshift=-0.15cm]s3.west);
    \draw[-Latex,dashed, very thick, gray!80, shorten <=0.3cm,shorten >=0.1cm]([yshift=-0.15cm]s3.north east)-- ([yshift=-0.15cm]s2.south west);
    \draw[-Latex, dashed, very thick, gray!80, shorten <=0.3cm,shorten >=0.7cm]([yshift=0.15cm]s2.east)-- ([yshift=0.15cm]h2.west);

    \node[server]       at ([yshift=2cm,xshift=1cm]s1)     (ctrl)    {};
    \node               at ([xshift=-2cm]ctrl)  (ctrll)   {SDN controller};
        
    \draw[thick, dotted, gray] (s1.north)--(ctrl.south east);
    \draw[thick, dotted, gray] (s2.north)--(ctrl.south west);
    \draw[thick, dotted, gray] ([xshift=0.1cm]s3.north east)--(ctrl.south);
    \draw[thick, dotted, gray] (s4.north west)--(ctrl.south);
        
    \node[rectangle,pattern=horizontal lines, minimum width=0.5cm, minimum height=1cm]    at([yshift=1.5cm, xshift=-2cm]ctrl)    (app1)  {};
    \node[] at ([xshift=-1cm]app1)    (app1l) {App. 1};
    \node[] at ([xshift=2cm]app1)   (dots)  {\huge\ldots};
    \node[rectangle,pattern=horizontal lines, minimum width=0.5cm, minimum height=1cm]    at([xshift=2cm]dots)    (app2)  {};
    \node[] at ([xshift=1cm]app2)    (app2l) {App. $N$};
        
    \draw[thick, dotted, gray] (app1.south)--(ctrl.north);
    \draw[thick, dotted, gray] (app2.south)--(ctrl.north);
\end{tikzpicture}

%% file: exact.tex
In Section~\ref{sec:intro}, we described the problem we are interested in solving. In order to avoid ambiguity from the description of our problem statement, we present the formal definition of the objects involved. Such definitions and notations are used throughout the rest of the paper. 

\begin{definition}\label{def:net} (Network) A network $N$ is a directed and weighted graph $(V,E,c)$, where:
	\begin{itemize}
			\item $V$ is a set of nodes;
			\item $E\subseteq V\times V$ is a set of edges (ordered pairs of nodes);
			\item $c: E\to \mathds{Z}_{+}$ is a bandwidth capacity function. 
	\end{itemize}
\end{definition}

As an example, for the network presented in Fig.~\ref{fig:ex_topo}, the model representing it is $(V,E,c)$, where $V=\{N_1,N_2,N_3,N_4\}$, $E=\{(N_1,N_2),(N_1,N_3),(N_1,N_4),(N_2,N_3),(N_2,N_4),(N_3,N_4),(N_2,N_1),(N_3,N_1),(N_4,N_1),(N_3,N_2),(N_4,N_2),(N_4,N_3)\}$, and $c$ is defined as follows:
\[
	c(e)=\begin{cases}
				1, & \text{if } e \in \{(N_2,N_1),(N_2,N_4),(N_3,N_4)\}\\
        		3, & \text{if } e \in \{(N_3,N_1),(N_2,N_3)\}\\
        		4, & \text{if } e \in \{(N_4,N_1)\}\\
        		2, & \text{otherwise.}
        	\end{cases}\
\]

Without loss of generality, we assume that the network is connected (meaning all nodes are reachable from any node in the network); otherwise each connected component can be considered an independent network. Following the traditional notions in graphs, we formally consider a path in a graph as a sequence of edges, i.e., a path $p\in E^*$. As an example, the path $N_1\rightarrow N_2\rightarrow N_3\rightarrow N_4$ is formally the sequence of edges $(N_1,N_2)(N_2,N_3)(N_3,N_4)$. We note that, for convenience, we denote the source of an edge $e$ as $src(e)$, and its destination as $dst(e)$. Likewise, we denote the $i$-th edge of a path $p$ as $p_i$. Similarly, we denote the length of a path $p$ as $|p|$.

In our work, we are interested in centralized admission and routing in so-called transportation (or flow) networks. In such networks, nodes can receive and send flows, however, the capacity on the links cannot be exceeded. Informally, a (network) flow can be considered as ``a set of packets passing an observation point in the network during a certain time interval. All packets belonging to a particular flow have a set of common properties.'', following RFC 3917 \cite{RFC3917}. It is important to note that for the scope of our work, the common characteristics of the packets are found in the packet header (as those described in RFC 3917), however, the approach and formalism can be easily extended for considering some parts in the payload of the network packet (at least for non-encrypted protocols). Formally, a network flow is defined as follows.

\begin{definition}\label{def:flow} (Flow) A network flow $f$ in a network $N=(V,E,c)$ is a four-tuple, $(s,d,b,h)$. $s,d\in V$, are source ($s$) and destination ($d$) nodes for the flow, $b\in\mathds{Z}_{+}$ is the required bandwidth\footnote{We assume measured in a proper data transfer rate, e.g., megabytes per second.} of the flow, and finally $h\in\{0,1\}^k$ is the packet header (of $k$ bits) of the packets belonging to the flow, containing the common characteristics that distinguish it (from other flows). 
\end{definition}

For convenience, for a given flow $f=(s,d,b,h)$, we denote the source of that flow as $s(f)$, the destination as $d(f)$, the required bandwidth capacity as $b(f)$, and finally, the header as $h(f)$. Consider a flow with an eight-bit header $f=(N_1,N_2,2,00110011)$ for the network $N$, as shown in the example. The goal is to admit and send the flow with the characteristics encoded in $00110011$ from the source $N_1$ to the destination $N_2$, the bit rate at which the network packets with such characteristics arrive from $N_1$ is two (assume a proper bit rate unit). In order to be able to admit and send that flow, a path must be assigned to it; for example, a valid assignment for the flow $f$ is the path $(N_1,N_3)(N_3,N_4)(N_4,N_2)$. We note that the path $p$ can be empty (denoted $p=\epsilon$), and thus, it represents not admitting the flow (dropping packets at the source node). This is the general case for assigning network flows through a network. However, as previously stated, we focus on the prioritization of flows. Thus, there is an intrinsic notion of priority between flows. This priority can be a combination of different characteristics of the flows, for example, source, destination, type of application, etc. Without loss of generality we assume that all information which is relevant to distinguish between one flow and another can be found in the flow's header. For that reason, we formally define the concept of priority simply as:

\begin{definition}\label{def:prio} (Priority) A priority for the header of the packets belong to a flow $h(f)$ is a function $\mathcal{P}: \{0,1\}^k\to \mathds{Z}_{+}$.
\end{definition}

We note that the priority function maps a packet header to a given priority number. We assume that \emph{a larger priority value implies the more important the flows} with such header are to be considered. Having all the previous notions, we can formally state our problem of interest. 

\paragraph*{The Priority Flow Admission and Routing~(PFAR) problem} Given a network $N=(V,E,c)$, a finite set of flows $F$, and a priority function $\mathcal{P}$, find a path for each flow (mapping solution) $\mathcal{S}: F\to E^*$, such that: i) the paths are simple and valid for the flows and network; ii) the bandwidth of the allocated flows does not exceed the bandwidth capacity of each link; and finally iii) the value of the priority of the allocated flows is maximal. These three conditions are hereafter referred to as \emph{the properties of a PFAR solution}.

That a path $p$ is simple and valid for a given flow $f=(s,d,b,h)$ and network $N=(V,E,c)$ means that there are no repeated nodes in the path, and valid means that if $p\not=\epsilon$ (if it is not an empty path, i.e., the flow is to be dropped), $src(p_1)=s\wedge dst(p_{|p|})=d\wedge \forall i\in\{1,\ldots,|p| - 1\} dst(p_i)=src(p_{i+1})\wedge  \forall j\in\{1,\ldots,|p|\} p_j\in E$, that is, the path starts at $s$, finishes at $d$, it is connected, and respects the underlying network connectivity. That the bandwidth of the allocated flows does not exceed the capacity of each link is a restriction on each link, i.e., $\forall e\in E \sum_{f\in F^e} b(f) \leq c(e)$, where $F^e=\{f|f\in F\text{ and } \exists i\;\; \mathcal{S}(f)_i = e\}$, that is the sum of the required bandwidth of set of flows whose mapping (solution) transverses the link $e$ must be less or equal than the capacity of the edge $e$. Finally, that value of the priority of the allocated flows is maximal means that for any valid mapping of flows to paths $\mathcal{M}: F\to E^*$, $\mathcal{S}$ is optimal, i.e., $\sum_{f\in F^\mathcal{S}}\mathcal{P}(h(f))\geq \sum_{f\in F^\mathcal{M}}\mathcal{P}(h(f))$, where $F^\mathcal{S}=\{f|f\in F \text{ and } \mathcal{S}(f)\not=\epsilon\}$, and $F^\mathcal{M}=\{f|f\in F \text{ and } \mathcal{M}(f)\not=\epsilon\}$, that is, the sum of the priority of the set of admitted flows by the solution $\mathcal{S}$ is greater or equal than the sum of the priority of the set of admitted flows for any other valid mapping.

The PFAR problem is a typical combinatorial optimization problem, and thus, we present a straightforward formulation of the PFAR problem as an ILP program.

\subsection{An exact PFAR solution using an ILP program}
To solve the PFAR problem using ILP, a simple algorithm can be constructed. This algorithm is shown in Algorithm~\ref{algo:exact}. The notation for the PFAR problem has been previously presented. However, the particular notations for the ILP problem are presented in Table~\ref{tab:notation}. The ILP formulation for a given PFAR problem instance $\langle N=(V,E,c), F=\{(s_1,d_1,b_1,h_1)\ldots,(s_{|F|},d_{|F|},b_{|F|},h_{|F|})\}, \mathcal{P}\rangle$ is presented in Equation~\ref{eq:gen_ilp}. 

\begin{algorithm}[!htb]
    \SetKwInOut{Input}{input}\SetKwInOut{Output}{output}\SetKw{KwBy}{by}
    \small
    \DontPrintSemicolon
    \Input{A network $N=(V,E,c)$, a set of flows $F=\{(s_1,d_1,b_1,h_1)\ldots,(s_{|F|},d_{|F|},b_{|F|},h_{|F|})\}$, and a priority function $\mathcal{P}$}
    \Output{A mapping $\mathcal{S}: F\to E*$, with the properties of a PFAR solution}
    \textbf{Step 1}: Construct an ILP program for the PFAR instance $\langle N,F,\mathcal{P}\rangle$ as shown in Equation~\ref{eq:gen_ilp} (see Table~\ref{tab:notation} for the ILP notation).\;
    \textbf{Step 2}: Solve the ILP program, and obtain the solution for variables $\rho_{i,m}$.\;
    \textbf{Step 3}: Set $\mathcal{S} \leftarrow \emptyset$\;
    \textbf{Step 4}: 
    \ForEach{$i \in \{1,\ldots,|F|\}$}
    {
    	Set $p\leftarrow\epsilon$ \tcp{Assign to $p$ the empty path}
    	
    	\If{$\exists \rho_{i,m} = 1$}
    	{
    		Set $p\leftarrow paths(f_i)_m$ \tcp{Set the $m$-th path for the $i$-th flow}
    	}

    	Set $\mathcal{S}\leftarrow\mathcal{S}\cup\{(f_i,p)\}$ \tcp{Add the mapping of the path $p$ to flow $f_i$}
    }
    \Return{$\mathcal{S}$}
    \caption{PFAR exact solution using ILP}\label{algo:exact}
\end{algorithm}

\begin{table}[!htb]
	\centering
	\begin{tabular}{|l|p{5cm}|}
		\hline
		\textbf{Symbol} 	& \textbf{Meaning} \\ \hline
		$f_i$           		& The $i$-th flow in lexicographical order of $F$\\ \hline
		$\alpha_i$           	& A variable denoting if the $i$-th flow should be admitted ($\alpha_i=1$) or dropped ($\alpha_i=0$)  \\ \hline
		$\mathbf{paths}(f_i)$	& A set of (simple and valid) paths for flow $f_i$ \\ \hline
		$\mathbf{paths}(f_i)_m$	& The $m$-th path in lexicographical order for the set of paths for flow $f_i$ \\ \hline
		$\rho_{i,m}$          	& A variable denoting if the $i$-th flow should be routed via the $m$-th path for $f_i$ (in lexicographical order of $\mathbf{paths}(f_i)$) \\ \hline
		$\varepsilon_{i,j,l}$	& A variable denoting if the $i$-th flow should be routed through the edge $(j,l)$ \\ \hline
        $\mathbf{1_{i}}((j,l))$   & A characteristic function, indicating if edge $(j,l)$ belongs to the any path for the $i$-th flow\\\hline
	\end{tabular}
	\caption{ILP program notation}
	\label{tab:notation}
\end{table}

\begin{strip}
\center
\begin{equation}\label{eq:gen_ilp}
	\begin{aligned}
        &\text{\textbf{maximize}} \sum_{i=1}^{|F|}\mathcal{P}(h(f_i))*\alpha_i \\
        &\text{\textbf{subject to}:}\\		
        &\sum_{m=1}^{|\mathbf{paths}(f_i)|}\rho_{i,m} = \alpha_i; & \forall i \in\{1,\ldots,|F|\} \text{ (i)}\\
        &\sum_{(j,l)\in\mathbf{paths}(f_i)_m}\varepsilon_{i,j,l} \geq |\mathbf{paths}(f_i)_m|*\rho_{i,m}; &\forall i \in\{1,\ldots,|F|\} \forall m \in\{1,\ldots,\mathbf{paths}(f_i)\}\text{ (ii)}\\
        &\sum_{i}^{|F|}b(f_i)*\varepsilon_{i,j,l} \leq c((j,l)); &\forall (j,l)\in E\text{ (iii)}\\
        &\sum_{(j,l)\in E \text{ and } \mathbf{1_{i}}((j,l))=0} \varepsilon_{i,j,l} = 0; &\forall i \in\{1,\ldots,|F|\} \text{ (iv)}\\
        &\alpha_i \in\{0,1\}; &\forall i \in\{1,\ldots,|F|\}\text{ (v)}\\
        &\rho_{i,m}\in\{0,1\}; &\forall i \in\{1,\ldots,|F|\} \forall m \in\{1,\ldots,\mathbf{paths}(f_i)\}\text{ (vi)}\\
        &\varepsilon_{i,j,l}\in\{0,1\}; &\forall i \in\{1,\ldots,|F|\} \forall (i,j)\in E \text{ (vii)}\\
    \end{aligned}
\end{equation}
\end{strip}

The key idea behind the algorithm is the construction of the ILP program. Essentially, the objective is to maximize the sum of the priorities of the admitted flows ($\alpha_i$). However, in order to admit a flow, one (and only one) path must be chosen ($\rho_{i,m}$), according to restriction (i). This path can be chosen only if all edges involved in the path ($\varepsilon_{i,j,l}$) can route the flow in question, as per restriction (ii). In order to determine if an edge can route the flow in question, the sum of the required bandwidth of all the candidate flows to be routed through that edge must be less or equal than the bandwidth capacity of the edge, as shown in restriction (iii). Finally, edges that cannot route a flow since they do not belong in any of the paths for that flow must be set to zero ($\varepsilon_{i,j,l}$), according to restriction (iv), and all variables involved are binary, according to restrictions (v), (vi), and (vii). Note that a given edge can route a flow in multiple paths, and thus restriction (ii) is not a strict equality, to be able to have edges that route a flow in other paths. Once having the solution to this ILP problem, constructing the paths from the edges that can route the flows is straightforward. As shown in Step 4 of the algorithm, if there a path for a flow, this path is the proper mapping for the flow in question. As an example, the ILP program constructed for the network shown in Fig.~\ref{fig:ex_topo} and the flows shown in Table~\ref{tab:ex_flows} can be found in Appendix~\ref{sec:ap_ilp_prog}. The correctness of Algorithm~\ref{algo:exact} is proven by the following statement.

\begin{proposition}\label{stm:algo_correct}
Given A PFAR instance $\mathcal{I}=\langle N=(V,E,c),F,\mathcal{P}\rangle$, $\nexists \mathcal{M}:F\to E*\;\;\sum_{f\in F^\mathcal{M}}\mathcal{P}(h(f))> \sum_{f\in F^\mathcal{S}}\mathcal{P}(h(f))$, where $\mathcal{S}$ is the solution to $\mathcal{I}$ provided by Algorithm~\ref{algo:exact}, both $\mathcal{M}$, $\mathcal{S}$ map to paths that are simple and valid for $N$ and $F$, and the bandwidth of the allocated flows does not exceed the bandwidth capacity of each link.
\end{proposition}

\begin{proof}
By contradiction. Assume that there exists such an $\mathcal{M}$ that $\mathcal{M}:F\to E*\;\;\sum_{f\in F^\mathcal{M}}\mathcal{P}(h(f))> \sum_{f\in F^\mathcal{S}}\mathcal{P}(h(f))$. As the solution $\mathcal{S}$ was constructed using the ILP shown in Equation~\ref{eq:gen_ilp}, we know that at most a single path is chosen for a flow due to restriction (i), and that the bandwidth of the allocated flows are less or equal than the bandwidth capacity of each link due to restriction (ii) and (iii). Furthermore, we know that $\sum_{i=1}^{|F|}\mathcal{P}(h(f_i))*\alpha_i$ is maximal. Additionally, we know that for any $\alpha_i\not=0 \implies \exists \rho_{i,m}\not=0$. According to Algorithm~\ref{algo:exact}, for all these flows a non-empty path is associated. Thus, $\sum_{f\in F^\mathcal{S}}\mathcal{P}(h(f))$ can be expressed as $\sum_{i=1}^{|F|}\mathcal{P}(h(f_i))*\alpha_i$, as for $\alpha_i=0$, empty paths are mapped, and furthermore, do not participate in the sum. As $\sum_{i=1}^{|F|}\mathcal{P}(h(f_i))*\alpha_i$ is maximal, for paths that are simple and valid for $N$ and $F$ and the bandwidth of the allocated does not exceed bandwidth capacity of each link, so is $\sum_{f\in F^\mathcal{S}}\mathcal{P}(h(f))$. Thus, a valid mapping $\mathcal{M}$ such that $\sum_{f\in F^\mathcal{M}}\mathcal{P}(h(f))> \sum_{f\in F^\mathcal{S}}\mathcal{P}(h(f))$ cannot exist, arriving to a contradiction.
\end{proof}

\paragraph*{Discussion -- on limiting the number of paths per flow} Indeed, the ILP formulation finds the exact solution for the PFAR problem. Nonetheless, as shown in Equation~\ref{eq:gen_ilp}, the general ILP form has many restrictions and variables that should be added to the problem for all paths for each flow. In the worst-case scenario, for a complete graph (full-mesh network) there are effectively $(|V| - 2)!\sum_{i=1}^{|V| - 2}\frac{1}{i!}$ paths for a single flow; note that the sum is always greater than or equal to one. For that reason, the size of the problem can rapidly grow (factorial growth) and the problem can become computationally expensive, further, finding a solution can be unfeasible in real time. There are different possibilities on how to limit the number of paths. Practically, it is often not desirable to have very long paths; mostly since the network packets may get long delays, and furthermore, the longer the path, the more edges (and thus resources) a flow occupies. Therefore, an easy way to limit the number of paths is to limit the length of the paths in question. For that reason, in our experimental section (Section~\ref{sec:experim}) we follow this approach. Other methods for limiting the number of paths are possible. However, the choice highly depends on the specific application, since not considering certain paths may imply loosing possibilities to route priority traffic that may be potentially admissible.

\paragraph*{Discussion -- on strict and non-strict prioritization of flow admission and routing} According to the PFAR problem formulation, flows with higher priority can be dropped, if the sum of priorities of lower priority flows surpasses that of the flow with higher priority. We formulated the problem in such a way as we consider this a desirable feature, since, few high priority flows can occupy all the network resources, and thus, a strict prioritization is not always desirable. Furthermore, by correctly assigning the priorities it can be guaranteed how many lower priority flows are necessary to overcome a higher priority flow. For example, as shown in Table~\ref{tab:ex_flows}, there is a factor of ten between each priority; that implies that for flows with a priority of 1000, effectively 101 flows of priority ten must be admitted in order to drop a single flow with priority 1000. However, we do recognize the interest of a strict flow prioritization. A simple pre-processing is necessary to use our proposed approaches. For a given instance $\langle N,F,\mathcal{P}\rangle$, first, group flows according to priority. That is to obtain non-empty sets with the different priorities, such that $F_n=\{f|f\in F \text{ and } \mathcal{P}(f)=n\}$. For each $F_n$, and in decreasing order of $n$, create a new instance $\langle N,F_n,\mathcal{P}\rangle$, and solve using our proposed approach. The advantage of such prioritization scheme is that it effectively reduces the number of flows per instance, and thus, the complexity of finding a solution. As discussed in this and the next section, the complexity of the proposed approach can be an issue.

%% file: heuristic.tex

In the previous section, we present an exact solution to the PFAR problem. However, the problem can yield instances for which the exact solution can only be computed in exponential time with respect to the number of nodes and flows of the problem instance\footnote{Note that the proof of the problem's complexity is outside the scope of this work, nonetheless, proving that the problem is NP-hard can be done through a reduction from one of the original 21 problems of Karp \cite{karp21}, the Knapsack problem (see \cite{vazirani2013approximation} for the proper formulation).}. For large instances (e.g., 50 nodes and 5000 flows), this is rather unfeasible, for runtime traffic management, especially for highly variable traffic. For that reason, we turn our attention to heuristic approaches. We focus on a solution of the PFAR problem that is still valid and respects the bandwidth capacity of the links, however, we sacrifice the optimality for the sake of computational time; this is a usual trade-off for such hard problems. 

We focus on GAs, as optimization problems like the PFAR problem are good candidates for employing such heuristic approaches. We discuss the particular characteristics of a GA tailored for the PFAR problem. 

\paragraph*{Individual representation} The natural representation of an individual for an ILP problem is an array of values, where each position represents one variable of the problem. As all variables in the problem are binary, a simple bit-string representation is well-suited for the genotype of the individuals. In order to better illustrate this concept, consider individuals contain only the $\rho$ variables of the ILP problem formulation. An individual $I$ has the following form: 

\begin{small}
	\[I=\underbrace{\overbrace{\rho_{1,1} \ldots \rho_{1,|paths(f_1)|}}^{|paths(f_1)|} \ldots \overbrace{\rho_{|F|,1}\ldots \rho_{|F|,|paths(f_{|F|})|}}^{|paths(f_{|F|})|}}_{\sum_{i=1}^{|F|}|paths(f_i)|}\]
\end{small}

For the particular example shown in Fig.~\ref{fig:ex_topo}, and Table~\ref{tab:ex_flows}, there are five possible paths for each flow, and thus, all individuals have a total length of 20; the individual whose genotype (bit-string) representation is $I=10000100001000010000$,  has a phenotype representation of assigning each flow to its first path, i.e., $\rho_{i,1}=1\;\forall i\in\{1,2,3,4\}$.

If considering an individual as a bit-string of all the variables involved in the ILP problem, many individuals can represent unfeasible solutions. For example, if the bit-string in the position corresponding to admitting flow $i$ ($\alpha_i$) is equal to one but, all bits corresponding to which route to choose are zero, this is an unfeasible solution. To avoid having such fragile relationships in the individuals, different strategies have been implemented. The first is to \emph{consider only the route variables as the chromosome representation}. The reason is that these variables are enough to infer the proper value of other variables. For example, if $\rho_{i,m}=1$ it implies that all the edge variables corresponding to the $m$-th path for the $i$-th flow ($\varepsilon_{i,j,l}$) must be assigned to one, and $\alpha_i=1$. Additionally, such representation allows a less space (memory) consuming population.

\paragraph*{Initialization} As previously discussed, one of the peculiarities of the PFAR problem is its fragile relationships and dependencies between the variables. For that reason, in order to seed the initial population we assign each flow to at most one randomly chosen path. Furthermore, as assigning many of the flows to paths can surpass the link capacities, we do this assignation with a very low probability (0.01\%); this allows the genetic operators (mutation and crossover), and the fitness function to correctly guide the evolutionary process, and find very fit individuals. For the initial population we chose 100 individuals. The reason is that for large instances, the genotype may be a very long bit-string, and thus, processing long individuals can be time consuming; having more individuals may reduce the time to explore more generations and correctly guide the population toward a good fitness. 

In order to improve the search for fit individuals, one possible strategy is to create chromosomes with a \emph{smart} initialization \cite{greedy_GA}. We set one individual from the initial population to contain the solution provided by a \emph{greedy algorithm}. The greedy algorithm sorts the flows in non-increasing order of their priority, and tries to find a path that can accept the flow in that order. Later, a mapping is done to assign the found paths to the corresponding flows. Finally, a chromosome is created to reflect the choice of the greedy algorithm. This process is straightforward; for each flow if the flow is assigned a non-empty path, set the bit corresponding to the path to 1. We note that, any other polynomial-time algorithms can be considered for the initial population.

\paragraph*{Selection} In order to measure the fitness of each of the individuals in the population, a fitness function is defined. As previously stated, the fitness function can be the objective of the ILP problem, i.e., $\sum_{i=1}^{|F|}\mathcal{P}(h(f_i))*\alpha_i$, however, the problem does not have direct knowledge of the $\alpha$ variables. Thus, if the individuals are guaranteed to have at most a single path assigned, the fitness function can be modified to $\sum_{i=1}^{|F|}\sum_{m=1}^{|paths(f_i)|}\mathcal{P}(h(f_i))*\rho_{i,m}$. Nevertheless, when assigning a flow to a given path, the genetic algorithm has no knowledge whether this assignment surpasses the capacity of some link, providing an unfeasible solution. In order to avoid such cases, inspired by reinforcement learning \cite{sutton2018reinforcement}, we propose that the fitness function is capable of providing negative rewards for individuals which assign a flow to a path which surpasses a link capacity. Then, the proposed fitness function is: $\sum_{i=1}^{|F|}\sum_{m=1}^{|paths(f_i)|}\mathcal{P}(h(f_i))*\rho_{i,m}*fits((V,E,c),F,i)$, where $fits$ is a function that equals 1 if the flow $i$ fits in the graph $(V,E,c)$ considering all previous $i-1$ flows that fit, and -1 otherwise. Note that, such definition highly depends on the order of the flows, however, we show that such definition effectively guides the search in Section~\ref{sec:experim}. Additionally, the penalization is cumulative for flows surpassing the link capabilities, if previous flows do not fit, they are considered as not admitted, and a negative reward is added to the fitness. With such formulation, placing flows in paths that surpasses a link capacity is discouraged as the fitness decreases when choosing such individuals, and therefore, such individuals are less likely to survive and pass to the next generations.

Having the fitness function defined the selection process is quite straightforward. First, individuals are sorted by non-increasing order of fitness. The fittest individuals are directly copied to the next generation; this is known as elitism, and serves the purpose of not losing the best found individuals. The number of individuals depends on the crossover rate (see \emph{genetic operators}), with the formula $PS*(1-CR)$, where $PS$ is the population size and $CR$ is the crossover rate. 

\paragraph*{Genetic operators (generation)} The genetic operators are designed with the goal of avoiding unfeasible solutions, considering the fragile relationships between variables. The key principle is to preserve at most one path per flow. Taking this into consideration, the \emph{mutation operator} first selects a flow to mutate, then sets all bits corresponding to route this flow via a path to zero. Then, with low probability (0.01\%) a random path is chosen for the selected flow. With such mutation operator, it is guaranteed that at most one path is selected to route a given flow.

The \emph{crossover operator} is performed on a number of individuals w.r.t. the formula $CS=PS*CR$, where $CS$ is the crossover size (the individuals to reproduce), $PS$ is the population size and $CR$ is the crossover rate. Similarly, the mutation operator is applied to a number of individuals w.r.t. the formula $MS=CS*MR$, where $MS$ is the mutation size (the individuals to mutate), and $MR$ is the mutation rate. The rates start at $MR=0.1$ and $CR=0.9$, during the generations, the mutation rate increases while the crossover rate decreases while maintaining the constraint $MR+CR=1$, as proposed in \cite{dynamic_rates}. Nonetheless, as the genetic algorithm terminates in a fixed time (see \emph{termination}), the rates are adapted according to the elapsed time, and the estimated number of remaining generations; such time-adaptive rate scheme is a novel contribution and may help other GAs. 

\paragraph*{Termination} Indeed, the PFAR problem is a bounded optimization problem, i.e., the optimal cannot be greater than $\sum_{i=1}^{|F|}\mathcal{P}(h(f_i))$ as that implies that all flows are admitted. Thus, two criteria are chosen for terminating the genetic algorithm. The first is to obtain an individual with fitness equal to the optimal value, which happens only for instances in which the optimal is the maximal possible value. The second is if the generation process has been executed for a constant time. The motivation behind this is to find the best possible solution in an acceptable time. In our experiments we set this constant time to ten seconds, as we are interested in re-configuring the network with such time intervals. However, this constant can be changed. The algorithm repeats the generation step until one of the two conditions for termination is met.

%% file: experiments.tex
As previously discussed, there are theoretical advantages and disadvantages to both the exact and heuristic approaches presented in previous sections. On the one hand, the exact solution based on ILP may provide an optimal and exact solution. However, this approach is computationally hard, and may not be able to provide a solution in a reasonable amount of time, for large instances (w.r.t. the number of nodes and flows in the network). On the other hand, the heuristic solution based on GAs is guaranteed to provide the best known solution in a constant time. Nonetheless, there is no guarantee that the best solution is good. Furthermore, it is interesting to compare how far the solution is from the optimal. For those reasons, in this section we present an experimental evaluation. Before presenting the obtained results, we present the particularities of the generated instances.

\paragraph*{Network topology / bandwidth capacity} As both approaches are generic, and topology agnostic, the experimental evaluation can be done with any topology or bandwidth assignation. As we feel that completely aleatory instances may yield good but uninteresting results, we focused on networks with specific topologies and bandwidth relationships. First, due to our particular interest in such topologies, and second because we feel such topologies are widespread, and furthermore, they capture different topological arrangements. In order to better explain topologies of interest, we present an example of 11 nodes in Fig.~\ref{fig:exper_topo}. The topologies of interest are based on double-star networks. Such hierarchical networks are of general interest. Additionally, we include a full mesh sub-component, in order to provide a diverse sample. Finally, nodes in the first level are connected in a so-called bus configuration, and some nodes are isolated, connected only through the root of the topology. The bandwidth of the root node to the first level nodes is the highest ($L0$), while the return bandwidth for each link is at most half. Similarly, the bandwidth of the first level to the second level is inferior than the return links ($L1$), and finally, the smallest bandwidth is considered for nodes at the second level. This type of topology considers that traffic is concentrated at the root or much of the traffic passes through the root. In order to vary the passing bandwidth, a maximal constant for each level is defined, and a random fraction between one and one-tenth is assigned to each link. Such topologies can be generated via Algorithm~\ref{algo:topo_gen}. For our experimental evaluation we chose $L0=$30Mbps, $L1=$10Mbps, and $L2=$2Mbps.

\begin{figure}[!htb]
    \centering
    \includegraphics[width=\columnwidth]{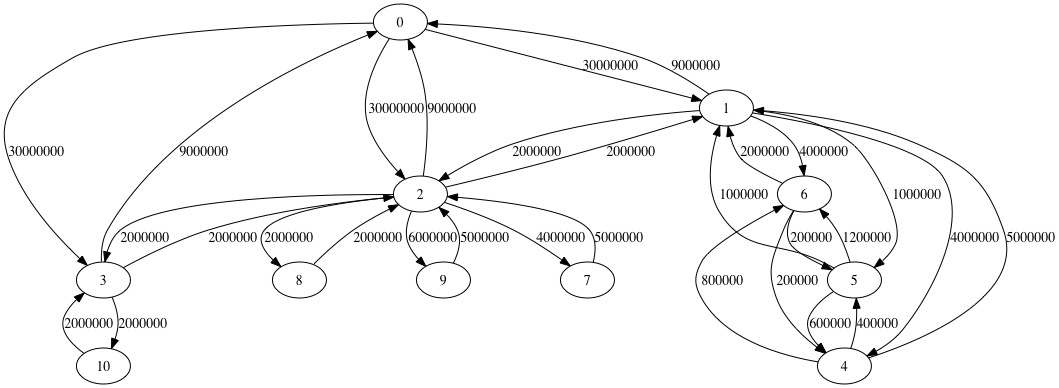}
    \caption{Example topology}
    \label{fig:exper_topo}
\end{figure}

\begin{algorithm}[!htb]
    \small
    \SetKwInOut{Input}{input}\SetKwInOut{Output}{output}\SetKw{KwBy}{by}
    \DontPrintSemicolon
    \Input{A number of nodes $N$, maximal bandwidth constants $L0$, $L1$, $L2$}
    \Output{A graph with the desired topology / bandwidth properties}
    \textbf{Step 0}: To have a double-star base, solve (for the positive root) the equation $N=n^2+n+1$, where $n$ is the number of children per node.\;
    \textbf{Step 1}: Create the graph $G$ with $N$ nodes\; 
    \textbf{Step 2}: Add the links $(0,i)$ and $(i,0)$ to $G$ with $\frac{L0}{rand(1,10)}$ and $\frac{L0}{2*(rand(1,10))}$ bandwidth, respectively $\forall i \in \{1,\ldots,n\}$\;
	\textbf{Step 3}: Add the links $(i,i+1)$ and $(i+1,i)$ to $G$  with $\frac{L1}{rand(1,10)}$ bandwidth $\forall i \in \{1,\ldots,n -1\}$\tcp{bus-like connections for nodes of the first level.}
	\textbf{Step 4}: Add the links $(i,j)$ and $(j,i)$ to $G$ with $\frac{L1}{rand(1,10)}$ and $\frac{L1}{2*rand(1,10)}$ bandwidth, respectively  $\forall i \in \{1,\ldots,n\} \forall j \in \{i*n+1,\ldots,i*n+n\}$\;
	\textbf{Step 5}: Add the links $(j,k)$ and $(k,j)$ to $G$ with $\frac{L2}{rand(1,10)}$ bandwidth $\forall j \in \{i*n+1,\ldots,i*n+n\} \forall k \in \{2*i,\ldots,3*i\}$\tcp{mesh connections for the first sub-cluster}
    \textbf{Step 6}: Add the links $(0,i)$ and $(i,0)$ to $G$ with $\frac{L0}{rand(1,10)}$ and $\frac{L0}{2*rand(1,10)}$ bandwidth, respectively $\forall i \in \{n^2 + 1,\ldots,N - 1\}$\tcp{Add the remainder nodes to the root}
    \Return{$G$}
    \caption{Topology generation}\label{algo:topo_gen}
\end{algorithm}

\paragraph*{Flow characteristics and distribution} In our experiments, it is interesting to generate many flows, and furthermore, to generate instances in which the network capacity is surpassed by those flows (n.b. this implies that optimization bound is not reachable). To generate such congestion, a simple approach has been taken. First, flows must always fit all links, therefore, their bandwidth requirements are always inferior to the smallest bandwidth capacity of the graph, that is $\frac{L2}{rand(1-10)}$. Additionally, to force our algorithms to always choose between at least two flows, flows have random size between  $\frac{L2}{2* rand(1-10)}$ and $ \frac{L2}{rand(1-10)} - 1$. The amount of flows to generate depends on the sum of all the outgoing capacity of a node. For each node, flows are generated to randomly and uniformly distributed destination of nodes until all outgoing capacity of the node is surpassed. By doing so, congestion per node is guaranteed. Furthermore, the entire network gets congested, since links are involved in many paths, as flows are uniformly distributed. This configuration generates hard instances, in which for a network of two nodes hundreds of flows are generated, and a choice between them must be made. As the size of instances increases, instances become harder.

Finally, the priority distribution of flows has been chosen in the following manner. First, the number of different priorities has been limited to five, as we feel in realistic environments a small number is enough to prioritize traffic. The priorities are 1, 10, 100, 1000, and 10000. The desired distribution of flows is $50\%$, $27\%$, $13\%$, $7\%$ and $3\%$, respectively. Note that, since flows are randomly generated, such distribution may differ slightly. 

\paragraph*{Experimental setup} The experiments have been executed on an Ubuntu 20.04 LTS, kernel 5.4.0-47 with 16GB of RAM and eight Intel Core Processor (Broadwell). All software modules have been developed using C++, and compiled with the GNU Compiler Collection. For solving the ILP program, the C++ interface of the Gurobi\footnote{Note that, the CPLEX solver has been also tested and its performance for our instances is inferior compared to Gurobi thus, justifying our choice.} \cite{gurobi} solver has been used. The instances were generated as previously described, the size of the instances varies between two and 50 nodes. 

\paragraph*{Experimental results} The experimental results are shown in Table~\ref{tab:results}. As can be seen, the optimal solution\footnote{Considering an optimality tolerance for the solver to provide faster results.} is often found in acceptable time, and furthermore the heuristic approach is well-motivated by the fact that the exact solution can sometimes be found only in very a long time (see rows corresponding to 34 and 48 nodes). Further, the results obtained by the heuristic solution are close-to-optimal. In order to better visualize the results, in Fig.~\ref{fig:ga_opt_ratio} we show the GA optimality ratio; likewise, we show the ILP running time in Fig.~\ref{fig:ilp_rt} (recall, the GA is configured to terminate in ten seconds or less).


\begin{table}[!htb]
    \centering
    \begin{tabular}{c|c|c|c}
        \hline
        \textbf{Instance}       & \textbf{Optimal}  & \textbf{ILP time(s)}  & \textbf{GA value} \\ \hline
        $|V|$ = 2 $|F|$ = 301 & 156251 & 0.019 & 156251 \\ \hline
        $|V|$ = 3 $|F|$ = 244 & 112070 & 0.018 & 112067 \\ \hline
        $|V|$ = 4 $|F|$ = 348 & 88656 & 0.031 & 88626 \\ \hline
        $|V|$ = 5 $|F|$ = 633 & 157947 & 0.37 & 157893 \\ \hline
        $|V|$ = 6 $|F|$ = 794 & 380215 & 0.31 & 380038 \\ \hline
        $|V|$ = 7 $|F|$ = 972 & 376780 & 0.54 & 375829 \\ \hline
        $|V|$ = 8 $|F|$ = 798 & 232019 & 0.66 & 231161 \\ \hline
        $|V|$ = 9 $|F|$ = 813 & 323810 & 0.59 & 323193 \\ \hline
        $|V|$ = 10 $|F|$ = 1146 & 381757 & 0.65 & 379072 \\ \hline
        $|V|$ = 11 $|F|$ = 1162 & 422800 & 4.6 & 419916 \\ \hline
        $|V|$ = 12 $|F|$ = 1466 & 588925 & 1.4 & 582441 \\ \hline
        $|V|$ = 13 $|F|$ = 1989 & 683404 & 3.2 & 646450 \\ \hline
        $|V|$ = 14 $|F|$ = 1627 & 609419 & 1.6 & 594272 \\ \hline
        $|V|$ = 15 $|F|$ = 2269 & 774025 & 3.4 & 742937 \\ \hline
        $|V|$ = 16 $|F|$ = 2091 & 871088 & 1.8 & 844333 \\ \hline
        $|V|$ = 17 $|F|$ = 1224 & 555164 & 1.3 & 550268 \\ \hline
        $|V|$ = 18 $|F|$ = 1607 & 669305 & 0.87 & 656240 \\ \hline
        $|V|$ = 19 $|F|$ = 1776 & 793718 & 1.8 & 783000 \\ \hline
        $|V|$ = 20 $|F|$ = 2584 & 882535 & 5 & 864960 \\ \hline
        $|V|$ = 21 $|F|$ = 2533 & 963391 & 6.1 & 945893 \\ \hline
        $|V|$ = 22 $|F|$ = 3888 & 1588351 & 3.8 & 1561413 \\ \hline
        $|V|$ = 23 $|F|$ = 2726 & 873393 & 7.5 & 839107 \\ \hline
        $|V|$ = 24 $|F|$ = 2898 & 1193431 & 5.6 & 1095777 \\ \hline
        $|V|$ = 25 $|F|$ = 2733 & 1039981 & 7.4 & 958294 \\ \hline
        $|V|$ = 26 $|F|$ = 2408 & 901821 & 19 & 845234 \\ \hline
        $|V|$ = 27 $|F|$ = 2569 & 832278 & 2.3 & 820566 \\ \hline
        $|V|$ = 28 $|F|$ = 2116 & 862648 & 12 & 840215 \\ \hline
        $|V|$ = 29 $|F|$ = 2709 & 1028725 & 9.9 & 980480 \\ \hline
        $|V|$ = 30 $|F|$ = 2325 & 795948 & 3.1 & 782428 \\ \hline
        $|V|$ = 31 $|F|$ = 3889 & 1498438 & 11 & 1434224 \\ \hline
        $|V|$ = 32 $|F|$ = 4625 & 1625792 & 8.9 & 1578978 \\ \hline
        $|V|$ = 33 $|F|$ = 3899 & 1375422 & 11 & 1358840 \\ \hline
        $|V|$ = 34 $|F|$ = 3191 & 1165834 & 220 & 1131384 \\ \hline
        $|V|$ = 35 $|F|$ = 4533 & 1560722 & 26 & 1383051 \\ \hline
        $|V|$ = 36 $|F|$ = 3076 & 976290 & 11 & 916303 \\ \hline
        $|V|$ = 37 $|F|$ = 4588 & 1694784 & 3.5 & 1691138 \\ \hline
        $|V|$ = 38 $|F|$ = 2721 & 1081105 & 17 & 1025966 \\ \hline
        $|V|$ = 39 $|F|$ = 3856 & 1563589 & 9.2 & 1458821 \\ \hline
        $|V|$ = 40 $|F|$ = 3810 & 1341606 & 26 & 1304439 \\ \hline
        $|V|$ = 41 $|F|$ = 3430 & 1286082 & 3.8 & 1254889 \\ \hline
        $|V|$ = 42 $|F|$ = 4561 & 1912342 & 4.9 & 1798727 \\ \hline
        $|V|$ = 43 $|F|$ = 6481 & 2297201 & 14 & 2230465 \\ \hline
        $|V|$ = 44 $|F|$ = 5519 & 1855886 & 15 & 1624787 \\ \hline
        $|V|$ = 45 $|F|$ = 5703 & 2102188 & 12 & 1911110 \\ \hline
        $|V|$ = 46 $|F|$ = 4926 & 1832282 & 39 & 1792867 \\ \hline
        $|V|$ = 47 $|F|$ = 4951 & 1891856 & 20 & 1800362 \\ \hline
        $|V|$ = 48 $|F|$ = 3617 & 1242983 & 490 & 1224566 \\ \hline
        $|V|$ = 49 $|F|$ = 4113 & 1547782 & 13 & 1464983 \\ \hline
        $|V|$ = 50 $|F|$ = 3518 & 1353846 & 34 & 1279134 \\ \hline
    \end{tabular}
    \caption{Experimental Results}
    \label{tab:results}
\end{table}

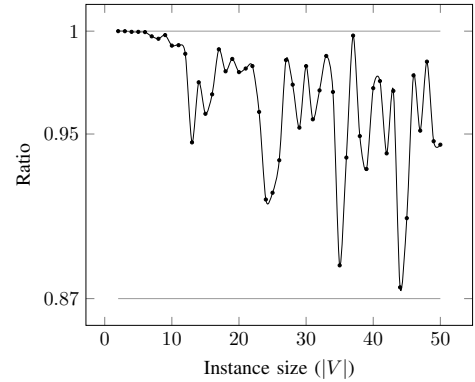
\begin{figure}[!htb]
    \centering
    \input{figures/GA_opt_ratio}
    \caption{GA optimality ratio}
    \label{fig:ga_opt_ratio}
\end{figure}

\begin{figure}[!htb]
    \centering
    \input{figures/ilp_rt}
    \caption{ILP vs. GA running time}
    \label{fig:ilp_rt}
\end{figure}
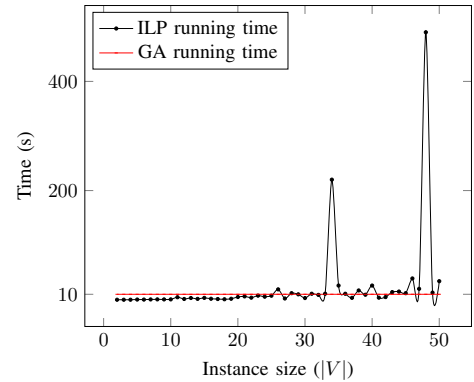

%% file: figures/GA_opt_ratio.tex
\begin{tikzpicture}[scale=.75]
    \begin{axis}[
        xlabel=Instance size ($|V|$),
        ylabel=Ratio,
        legend style={at={(0.02,0.98)}, anchor=north west},
        ytick = {0, 0.8, 0.87, 0.95, 1}
        ]
        
    \addplot[smooth,mark=*, mark size=0.8pt, black] plot coordinates 
    {
		(2,1)
		(3,0.999973231016329)
		(4,0.999661613427179)
		(5,0.999658113164542)
		(6,0.999534473916074)
		(7,0.99747598067838)
		(8,0.996302026989169)
		(9,0.998094561625645)
		(10,0.992966730145092)
		(11,0.99317880794702)
		(12,0.988990109097084)
		(13,0.945926567593985)
		(14,0.975145179260903)
		(15,0.959835922612319)
		(16,0.969285537167313)
		(17,0.991180984357775)
		(18,0.980479751383898)
		(19,0.986496463479472)
		(20,0.980085775634961)
		(21,0.981837073420864)
		(22,0.983040272584586)
		(23,0.96074390337454)
		(24,0.918173736060149)
		(25,0.921453372705848)
		(26,0.937252514634279)
		(27,0.985927778939249)
		(28,0.973995186912854)
		(29,0.953102140999781)
		(30,0.983013965736455)
		(31,0.957146041411123)
		(32,0.97120541865134)
		(33,0.98794406371281)
		(34,0.970450338555918)
		(35,0.886161020348275)
		(36,0.93855616671276)
		(37,0.997848693402817)
		(38,0.94899755342913)
		(39,0.93299517967957)
		(40,0.972296635524886)
		(41,0.97574571450343)
		(42,0.940588555812716)
		(43,0.97094899401489)
		(44,0.875477804132366)
		(45,0.909105179936333)
		(46,0.978488573265469)
		(47,0.951637968217454)
		(48,0.98518322454933)
		(49,0.946504740331649)
		(50,0.944814993728977)        
    };
    
    \addplot[domain=2:50, gray!80]{0.87};
    \addplot[domain=2:50, gray!80]{1};
    \end{axis}

\end{tikzpicture}

%% file: figures/ilp_rt.tex
\begin{tikzpicture}[scale=.75]
    \begin{axis}[
        xlabel=Instance size ($|V|$),
        ylabel=Time (s),
        legend style={at={(0.02,0.98)}, anchor=north west},
        ytick = {10, 200, 400}
        ]

    \addplot[smooth,mark=*, mark size=0.8pt, black] plot coordinates {
      	(2,0.019)
		(3,0.018)
		(4,0.031)
		(5,0.37)
		(6,0.31)
		(7,0.54)
		(8,0.66)
		(9,0.59)
		(10,0.65)
		(11,4.6)
		(12,1.4)
		(13,3.2)
		(14,1.6)
		(15,3.4)
		(16,1.8)
		(17,1.3)
		(18,0.87)
		(19,1.8)
		(20,5)
		(21,6.1)
		(22,3.8)
		(23,7.5)
		(24,5.6)
		(25,7.4)
		(26,19)
		(27,2.3)
		(28,12)
		(29,9.9)
		(30,3.1)
		(31,11)
		(32,8.9)
		(33,11)
		(34,220)
		(35,26)
		(36,11)
		(37,3.5)
		(38,17)
		(39,9.2)
		(40,26)
		(41,3.8)
		(42,4.9)
		(43,14)
		(44,15)
		(45,12)
		(46,39)
		(47,20)
		(48,4.9e+02)
		(49,13)
		(50,34)
    };
    \addlegendentry{ILP running time}
    
    \addplot[smooth,mark=-, mark size=0.8pt, red] plot coordinates {
      	(2,10)
		(3,10)
		(4,10)
		(5,10)
		(6,10)
		(7,10)
		(8,10)
		(9,10)
		(10,10)
		(11,10)
		(12,10)
		(13,10)
		(14,10)
		(15,10)
		(16,10)
		(17,10)
		(18,10)
		(19,10)
		(20,10)
		(21,10)
		(22,10)
		(23,10)
		(24,10)
		(25,10)
		(26,10)
		(27,10)
		(28,10)
		(29,10)
		(30,10)
		(31,10)
		(32,10)
		(33,10)
		(34,10)
		(35,10)
		(36,10)
		(37,10)
		(38,10)
		(39,10)
		(40,10)
		(41,10)
		(42,10)
		(43,10)
		(44,10)
		(45,10)
		(46,10)
		(47,10)
		(48,10)
		(49,10)
		(50,10)
    };
    \addlegendentry{GA running time}
    
    \end{axis}

\end{tikzpicture}

%% file: conclusion.tex
In this paper, we have proposed a novel priority admission and routing scheme by leveraging the centralized SDN capabilities. We have proposed an exact method to solve the proposed problem via an integer linear programming program. Due to the potentially high computational complexity of the exact approach, a heuristic approach has been proposed; this heuristic approach is based on genetic algorithms and incorporates a fitness function inspired by reinforcement learning. Our experimental results show that both approaches are pertinent.

Despite the promising results obtained in this work, many aspects are left for future work. First, since the exact approach does not guarantee a \emph{good} running time nor the heuristic approach guarantees an optimality approximation, it is interesting to study approximation algorithms for the priority admission and routing problem \cite{vazirani2013approximation}. Another direction which we plan to study is to further tune the genetic algorithm parameters. Additionally, it is interesting to study other heuristic approaches such as particle swarm optimization. Finally, another interesting perspective is to study other functional restrictions on the path assignation for the flows. 

%% file: ilp_example.tex
In this section, we present an example for a particular ILP instance (the running example presented in Fig.~\ref{fig:ex_topo} and Table~\ref{tab:ex_flows}) in Equation~\ref{eq:ex_ilp}. It is important to note that the ILP formulation largely depends on the resulting $\mathbf{paths}(f_i)$ function, which is a set of paths for a given flow $f_i$. For that reason, we present first the set of paths for each flow in Table~\ref{tab:paths}. Note that the paths for a flow depend on its source and destination. Therefore, the set of paths is repeated for many flows. Likewise, we consider that the paths are listed in the set's lexicographical order, i.e., the order that they appear listed. For example we refer to the first path for the first flow as $\rho_{1,1}=(N_1,N_2)$, $\rho_{1,2}=(N_1,N_3)(N_3,N_2)$, etc. Finally, note that the restrictions on the domain of the variables has been omitted, as all the variables are binary.

\begin{table}[!htb]
    \centering
    \begin{tabular}{|c|p{8cm}|}
        \hline
        \textbf{ID} & \textbf{paths}    \\ \hline
        1           & $\{(N_1,N_2),(N_1,N_3)(N_3,N_2),(N_1,N_4)(N_4,N_2),(N_1,N_3)(N_3,N_4)(N_4,N_2),(N_1,N_4)(N_4,N_3)(N_3,N_2)\}$\\ \hline
        2           & Same as the paths for flow with ID 1. \\ \hline
        3           & $\{(N_3,N_2),(N_3,N_1)(N_1,N_2),(N_3,N_4)(N_4,N_2),(N_3,N_1)(N_1,N_4)(N_4,N_2),(N_3,N_4)(N_4,N_1)(N_1,N_2)\}$  \\ \hline
        4           & Same as the paths for flow with ID 1.\\ \hline
    \end{tabular}
    \caption{Paths for flows}
    \label{tab:paths}
\end{table}

\begin{equation}
\label{eq:ex_ilp}
\begin{footnotesize}
\begin{aligned}
    & {\text{\textbf{maximize}}}	\;\; 10*\alpha_1 + 1000*\alpha_2+ \alpha_3 + 100*\alpha_4  \\
    & {\text{\textbf{subject to}:}}\\
    & \rho_{1,1}+\rho_{1,2}+\rho_{1,3}+\rho_{1,4}+\rho_{1,5} = \alpha_1\\
    & \rho_{2,1}+\rho_{2,2}+\rho_{2,3}+\rho_{2,4}+\rho_{2,5} = \alpha_2\\
    & \rho_{3,1}+\rho_{3,2}+\rho_{3,3}+\rho_{3,4}+\rho_{3,5} = \alpha_3\\
    & \rho_{4,1}+\rho_{4,2}+\rho_{4,3}+\rho_{4,4}+\rho_{4,5} = \alpha_4\\
    & \varepsilon_{1,1,2}\geq\rho_{1,1}\\
    & \varepsilon_{1,1,3}+\varepsilon_{1,3,2}\geq2*\rho_{1,2}\\
    & \varepsilon_{1,1,4}+\varepsilon_{1,4,2}\geq2*\rho_{1,3}\\
    & \varepsilon_{1,1,3}+\varepsilon_{1,3,4}+\varepsilon_{1,4,2}\geq3*\rho_{1,4}\\
    & \varepsilon_{1,1,4}+\varepsilon_{1,4,3}+\varepsilon_{1,3,2}\geq3*\rho_{1,5}\\
    & \varepsilon_{2,1,2}\geq\rho_{2,1}\\
    & \varepsilon_{2,1,3}+\varepsilon_{2,3,2}\geq2*\rho_{2,2}\\
    & \varepsilon_{2,1,4}+\varepsilon_{2,4,2}\geq2*\rho_{2,3}\\
    & \varepsilon_{2,1,3}+\varepsilon_{2,3,4}+\varepsilon_{2,4,2}\geq3*\rho_{2,4}\\
    & \varepsilon_{2,1,4}+\varepsilon_{2,4,3}+\varepsilon_{2,3,2}\geq3*\rho_{2,5}\\
    & \varepsilon_{3,3,2}\geq\rho_{3,1}\\
    & \varepsilon_{3,3,1}+\varepsilon_{3,1,2}\geq2*\rho_{3,2}\\
    & \varepsilon_{3,3,4}+\varepsilon_{3,4,2}\geq2*\rho_{3,3}\\
    & \varepsilon_{3,3,1}+\varepsilon_{3,1,4}+\varepsilon_{3,4,2}\geq3*\rho_{3,4}\\
    & \varepsilon_{3,3,4}+\varepsilon_{3,4,1}+\varepsilon_{3,1,2}\geq3*\rho_{3,5}\\
    & \varepsilon_{4,1,2}\geq\rho_{4,1}\\
    & \varepsilon_{4,1,3}+\varepsilon_{4,3,2}\geq2*\rho_{4,2}\\
    & \varepsilon_{4,1,4}+\varepsilon_{4,4,2}\geq2*\rho_{4,3}\\
    & \varepsilon_{4,1,3}+\varepsilon_{4,3,4}+\varepsilon_{4,4,2}\geq3*\rho_{4,4}\\
    & \varepsilon_{4,1,4}+\varepsilon_{4,4,3}+\varepsilon_{4,3,2}\geq3*\rho_{4,5}\\
    & 2*\varepsilon_{1,1,2}+2*\varepsilon_{2,1,2}+\varepsilon_{3,1,2}+2*\varepsilon_{4,1,2}\leq2\\
    & 2*\varepsilon_{1,2,1}+2*\varepsilon_{2,2,1}+\varepsilon_{3,2,1}+2*\varepsilon_{4,2,1}\leq1\\
    & 2*\varepsilon_{1,1,3}+2*\varepsilon_{2,1,3}+\varepsilon_{3,1,3}+2*\varepsilon_{4,1,3}\leq2\\
    & 2*\varepsilon_{1,3,1}+2*\varepsilon_{2,3,1}+\varepsilon_{3,3,1}+2*\varepsilon_{4,3,1}\leq3\\
    & 2*\varepsilon_{1,1,4}+2*\varepsilon_{2,1,4}+\varepsilon_{3,1,4}+2*\varepsilon_{4,1,4}\leq2\\
    & 2*\varepsilon_{1,4,1}+2*\varepsilon_{2,4,1}+\varepsilon_{3,4,1}+2*\varepsilon_{4,4,1}\leq4\\
    & 2*\varepsilon_{1,2,3}+2*\varepsilon_{2,2,3}+\varepsilon_{3,2,3}+2*\varepsilon_{4,2,3}\leq3\\
    & 2*\varepsilon_{1,3,2}+2*\varepsilon_{2,3,2}+\varepsilon_{3,3,2}+2*\varepsilon_{4,3,2}\leq2\\
    & 2*\varepsilon_{1,2,4}+2*\varepsilon_{2,2,4}+\varepsilon_{3,2,4}+2*\varepsilon_{4,2,4}\leq1\\
    & 2*\varepsilon_{1,4,2}+2*\varepsilon_{2,4,2}+\varepsilon_{3,4,2}+2*\varepsilon_{4,4,2}\leq2\\
    & 2*\varepsilon_{1,3,4}+2*\varepsilon_{2,3,4}+\varepsilon_{3,3,4}+2*\varepsilon_{4,3,4}\leq1\\
    & 2*\varepsilon_{1,4,3}+2*\varepsilon_{2,4,3}+\varepsilon_{3,4,3}+2*\varepsilon_{4,4,3}\leq2\\
    & \varepsilon_{1,2,1}+\varepsilon_{1,3,1}+\varepsilon_{1,4,1}+\varepsilon_{1,2,3}+\varepsilon_{1,2,4}=0\\
    & \varepsilon_{2,2,1}+\varepsilon_{2,3,1}+\varepsilon_{2,4,1}+\varepsilon_{2,2,3}+\varepsilon_{2,2,4}=0\\
    & \varepsilon_{3,2,3}+\varepsilon_{3,1,3}+\varepsilon_{3,4,3}+\varepsilon_{3,2,1}+\varepsilon_{3,2,4}=0\\
    & \varepsilon_{4,2,1}+\varepsilon_{4,3,1}+\varepsilon_{4,4,1}+\varepsilon_{4,2,3}+\varepsilon_{4,2,4}=0
\end{aligned}
\end{footnotesize}
\end{equation}

%% file: main.bbl
\begin{thebibliography}{10}
\providecommand{\url}[1]{#1}
\csname url@samestyle\endcsname
\providecommand{\newblock}{\relax}
\providecommand{\bibinfo}[2]{#2}
\providecommand{\BIBentrySTDinterwordspacing}{\spaceskip=0pt\relax}
\providecommand{\BIBentryALTinterwordstretchfactor}{4}
\providecommand{\BIBentryALTinterwordspacing}{\spaceskip=\fontdimen2\font plus
\BIBentryALTinterwordstretchfactor\fontdimen3\font minus
  \fontdimen4\font\relax}
\providecommand{\BIBforeignlanguage}[2]{{%
\expandafter\ifx\csname l@#1\endcsname\relax
\typeout{** WARNING: IEEEtran.bst: No hyphenation pattern has been}%
\typeout{** loaded for the language `#1'. Using the pattern for}%
\typeout{** the default language instead.}%
\else
\language=\csname l@#1\endcsname
\fi
#2}}
\providecommand{\BIBdecl}{\relax}
\BIBdecl

\bibitem{sdn}
\BIBentryALTinterwordspacing
Opennetworking, ``Software-defined networking: The new norm for networks,''
  \emph{ONF White Paper}, 2012. [Online]. Available:
  \url{https://www.opennetworking.org}
\BIBentrySTDinterwordspacing

\bibitem{googlesdn}
J.~Wanderer, ``Case study: The google sdn wan,'' \emph{Computing.co.uk},
  vol.~11, 2013.

\bibitem{LBSDN}
J.~{Li}, X.~{Chang}, Y.~{Ren}, Z.~{Zhang}, and G.~{Wang}, ``An effective path
  load balancing mechanism based on sdn,'' in \emph{2014 IEEE 13th
  International Conference on Trust, Security and Privacy in Computing and
  Communications}, 2014, pp. 527--533.

\bibitem{httpQoESDN}
E.~{Liotou}, K.~{Samdanis}, E.~{Pateromichelakis}, N.~{Passas}, and
  L.~{Merakos}, ``Qoe-sdn app: A rate-guided qoe-aware sdn-app for http
  adaptive video streaming,'' \emph{IEEE Journal on Selected Areas in
  Communications}, vol.~36, no.~3, pp. 598--615, 2018.

\bibitem{copbook}
C.~H. Papadimitriou and K.~Steiglitz, \emph{Combinatorial Optimization:
  Algorithms and Complexity}.\hskip 1em plus 0.5em minus 0.4em\relax Upper
  Saddle River, NJ, USA: Prentice-Hall, Inc., 1982.

\bibitem{dantzig}
G.~Dantzig, \emph{Linear programming and extensions}.\hskip 1em plus 0.5em
  minus 0.4em\relax Princeton university press, 1963.

\bibitem{gomory}
R.~E. Gomory, ``Outline of an algorithm for integer solutions to linear
  programs,'' \emph{Bulletin Amer. Math. Soc}, vol.~64, no.~5, pp. 275--278,
  1958.

\bibitem{gurobi}
G.~Optimization \emph{et~al.}, ``Gurobi optimizer reference manual,''
  \emph{URL: http://www. gurobi. com}, vol.~2, pp. 1--3, 2012.

\bibitem{cplex}
I.~I. Cplex, ``V12. 1: User’s manual for {CPLEX},'' \emph{International
  Business Machines Corporation}, vol.~46, no.~53, p. 157, 2009.

\bibitem{BnB}
\BIBentryALTinterwordspacing
A.~H. Land and A.~G. Doig, ``An automatic method of solving discrete
  programming problems,'' \emph{Econometrica}, vol.~28, no.~3, pp. 497--520,
  1960. [Online]. Available: \url{http://www.jstor.org/stable/1910129}
\BIBentrySTDinterwordspacing

\bibitem{BnC}
M.~Padberg and G.~Rinaldi, ``A branch-and-cut algorithm for the resolution of
  large-scale symmetric traveling salesman problems,'' \emph{SIAM Review},
  vol.~33, no.~1, pp. 60--100, 1991.

\bibitem{sipser}
M.~Sipser, \emph{Introduction to the Theory of Computation}.\hskip 1em plus
  0.5em minus 0.4em\relax Thomson Course Technology Boston, 2006, vol.~2.

\bibitem{goldbergGAs}
D.~E. Goldberg, \emph{Genetic Algorithms in Search, Optimization and Machine
  Learning}, 1st~ed.\hskip 1em plus 0.5em minus 0.4em\relax USA: Addison-Wesley
  Longman Publishing Co., Inc., 1989.

\bibitem{RFC3917}
J.~Quittek, T.~Zseby, B.~Claise, and S.~Zander, ``{RFC} 3917: requirements for
  ip flow information export: Ipfix,'' \emph{Published by Internet Engineering
  Task Force (IETF). Internet Society (ISOC) RFC Editor. USA. out}, 2004.

\bibitem{karp21}
R.~M. Karp, \emph{Reducibility among Combinatorial Problems}.\hskip 1em plus
  0.5em minus 0.4em\relax Boston, MA: Springer US, 1972, pp. 85--103.

\bibitem{vazirani2013approximation}
V.~V. Vazirani, \emph{Approximation algorithms}.\hskip 1em plus 0.5em minus
  0.4em\relax Springer Science \& Business Media, 2013.

\bibitem{greedy_GA}
Z.~{Wang}, H.~{Duan}, and X.~{Zhang}, ``An improved greedy genetic algorithm
  for solving travelling salesman problem,'' in \emph{2009 Fifth International
  Conference on Natural Computation}, vol.~5, 2009, pp. 374--378.

\bibitem{sutton2018reinforcement}
R.~S. Sutton and A.~G. Barto, \emph{Reinforcement learning: An
  introduction}.\hskip 1em plus 0.5em minus 0.4em\relax MIT press, 2018.

\bibitem{dynamic_rates}
A.~Hassanat, K.~Almohammadi, E.~Alkafaween, E.~Abunawas, A.~Hammouri, and
  V.~Prasath, ``Choosing mutation and crossover ratios for genetic
  algorithms—a review with a new dynamic approach,'' \emph{Information},
  vol.~10, no.~12, p. 390, 2019.

\end{thebibliography}
